\documentclass{birkjour}
\usepackage{fullpage}
\usepackage{graphicx}
\usepackage{amsmath}
\usepackage{amsthm} 
\usepackage{dsfont}
\usepackage{amssymb}
\usepackage{array}
\usepackage{multicol}
\usepackage[latin1]{inputenc}
\usepackage[T1]{fontenc}
\usepackage{MnSymbol}
\usepackage{multicol}
\usepackage{tikz}
\usepackage{tikz-cd}

\newtheorem{de}{Definition}[section]
\newtheorem{theo}{Theorem} [section] 
\newtheorem{prop}{Proposition}[section]
                                   
\newtheorem{lem}{Lemma}[section]
 
\newtheorem{coro}{Corollary}[section]

\newcommand{\Hcal}{\mathcal{H}}
\newcommand{\Kcal}{\mathcal{K}}

\newcommand{\Lcal}{\mathcal{L}}
\newcommand{\Rcal}{\mathcal{R}}
\newcommand{\Pcal}{\mathcal{P}}
\newcommand{\Acal}{\mathcal{A}}
\newcommand{\Scal}{\mathcal{S}}

\newcommand{\Ecal}{\mathcal{E}}
\newcommand{\Bcal}{\mathcal{B}}

\newcommand{\Prob}{\mathbb{P}}
\newcommand{\Fcal}{\mathcal{F}}
\newcommand{\Vcal}{\mathcal{V}}

\newcommand{\R}{\mathbb{R}}

\newcommand{\Zbb}{\mathbb{Z}}

\newcommand{\C}{\mathbb{C}}

\newcommand{\sca}[2]{\langle #1 ,#2\rangle}
\newcommand{\norm}[1]{\left\| #1 \right\|}
\newcommand{\proj}[2]{|#1\rangle\langle #2|}

\newcommand{\Tr}{\text{Tr}}

\numberwithin{equation}{section}

\begin{document}
\author{Ivan Bardet}
\title{Classical and Quantum Parts \\of the Quantum Dynamics: \\the Discrete-Time Case}
\address{Institut Camille Jordan, Universit\'e Claude Bernard Lyon 1, 43 boulevard du 11 novembre 1918, 69622 Villeurbanne, cedex France}
\email{ivan.bardet@math.univ-lyon1.fr}
\thanks{Work supported by ANR-14-CE25-0003 "StoQ" }
\begin{abstract}
In the study of open quantum systems modeled by a unitary evolution of a bipartite Hilbert space, we address the question of which parts of the environment can be said to have a "classical action" on the system, in the sense of acting as a classical stochastic process. Our method relies on the definition of the \emph{Environment Algebra}, a relevant von Neumann algebra of the environment. With this algebra we define the classical parts of the environment and prove a decomposition between a maximal classical part and a quantum part. Then we investigate what other information can be obtained via this algebra, which leads us to define a more pertinent algebra: the \emph{Environment Action Algebra}. This second algebra is linked to the minimal Stinespring representations induced by the unitary evolution on the system. Finally in finite dimension we give a characterization of both algebras in terms of the spectrum of a certain completely positive map acting on the states of the environment.
\end{abstract}
\maketitle

\section{Introduction}

In the Markovian interpretation of open quantum systems, the equation describing the evolution of a system is the sum of an Hamiltonian term and additional terms representing the noises introduced by the environment. In some cases it is possible to model the action of the environment on the system by a classical noise. This way the evolution of the system become the average over the noise of some unitary evolutions. In a more mathematical language, the evolution of the system is given by the average evolution of a Markov process taking value in the unitary group of the system.
\\For instance, the master equation describing the evolution of the state of the system can be written as a Schr\"odinger equation perturbed by some classical noises \cite{SBOM}. Such a model can be fruitful as it allows to borrow powerful tools from stochastic calculus to the study of open quantum systems. This approach has already been successfully used in the study of some two-levels systems \cite{A-B} and classical reducing \cite{Reb1}.
\\For such equations, K\"ummerer and Maassen show in \cite{K-M4} that the evolution can be dilated as a unitary evolution on the system and its environment, such that the environment appears as a Markov process. They called such dilations \emph{essentially commutative}. More particularly, they show that the three following assertions are equivalent when the system has a finite number of degree of liberty:
\begin{enumerate}
\item There exists a convolution semigroup of probability measures on the group of automorphisms of the set of linear map on the system, such that the evolution is given by the expectation of the convolution semigroup.
\item There exists an essentially commutative dilation of the Quantum Markov Semigroup.
\item The operator algebra generated by the evolution on the environment is commutative.
\end{enumerate}

In this article we have a different point of view, as we start directly from a unitary evolution between the system and its environment. More particularly we focus on the case of a one-step evolution, that is, when the evolution is given by a unitary operator on a tensor product of the system and the environment (a bipartite Hilbert space).
\\In this situation the equivalence between points 1. and 2. has already been proven by Attal and al. Indeed in \cite{ADP2} they show which unitary operators can be written in terms of classical noises emerging from the environment and characterize such noises: the \textit{obtuse random variables}. This provides a discrete analogue of an essentially commutative dilation. As proved in \cite{A-D}, in this case the evolution is equivalent to a random walk on the unitary group of the system. The starting point of this article is the equivalence between the two previous properties and the fact that the environment interferes with the system via a commutative algebra, which gives a more solid definition of a classical environment. 
\\On the other hand, some evolutions are understood to be typically quantum or non-commutative, although there is no clear definition of what this means. This is for instance the case for the Spontaneous Emission, where the one-step evolution is given by the following unitary operator on $\C^2\otimes\C^2$:

\begin{equation}
U_{\text{se}}=\begin{pmatrix}1&0&0&0\\0&\cos{\theta}&-\sin{\theta}&0\\0&\sin{\theta}&\cos{\theta}&0\\0&0&0&1\end{pmatrix}.
\label{eqSpontemiss}
\end{equation}

Our goal is to properly define the two situations mentioned above, that is, to make the distinction between a classical and a purely quantum environment. Thus the first goal of this article is to characterize, in terms of an operator algebraic framework, which unitary evolutions can be said to have a classical or purely quantum environment. Furthermore, we want to give a partial answer to the question: How far is the environment from being classical or purely quantum?
\\We answer this question by defining a relevant von Neumann algebra of the environment, that we call the \emph{Environment Algebra}. The environment is classical, or commutative in our definition, if this algebra is commutative. Moreover with the help of this algebra we are able to define and identify the classical parts of the environment, that is the parts where the dynamics reduces to an evolution with classical environment. In the same way we can say that an environment is quantum, or far from being classical, if it has no classical part. It naturally leads us to prove a decomposition of the environment between a classical and a quantum part. 

The environment algebra mentioned above allows to characterize whenever the unitary evolution can be written with classical noises. A natural question is then what other properties of the evolution can be obtained from it. However it appears that this algebra is "to big" to provide finer results. For instance in some cases a purely quantum environment can nonetheless lead to an evolution of the system driven by classical noises: a quantum environment can have a classical \emph{action} on the system. This comes from the fact that the system does not see evolutions occurring on the environment only. Consequently two evolutions on the bipartite system can lead to the same one when restricted to the system. We define the \emph{Environment Action Algebra} as the relevant algebra in this context. Two unitary operators that have the same action on the system will have the same Environment Action Algebra.
\\Under simple hypotheses, we prove that it is always possible to find a unitary operator on the bipartite system which leads to the same evolution on the system, but whose Environment Action Algebra coincide with its Environment Algebra. In this sense, it is always possible to restrict the study to the Environment Algebra. Then, in order to illustrate its usefulness, we show a link between the Environment Action Algebra and minimal Stinespring representations.
\\Finally, we characterize both algebras in terms of the spectrum of a completely positive map, providing a practical way to determine both algebras.

We insist on the fact that in this article the term \emph{classical} means that we deal, at least implicitly, with a \emph{classical probability space}. Later it will become clear that the mathematical meaning of this word is rather \emph{commutative}, as in \emph{commutative algebra}. In our discussions we shall use both terms without distinction but we will prefer the latter when stating mathematical results.
\\Note also that the classical noises we are talking about are different in nature from the ones emerging in the context of quantum trajectories (\cite{B-G1} \cite{Pel1} \cite{Pel2}). In the latter case, classical noises appears because of a continuous monitoring of some observable of the environment. In our case no observation is performed and consequently the emergence of classical noises has to be interpreted as a manifestation of the classical nature of the action of the environment on the system.

This article is structured as follow. In Section \ref{sect1} we focus on the definition of the Environment Algebra and the subsequent decomposition of the environment between a classical and quantum part. We give some examples and we study the particular case where the evolution is given by an explicit Hamiltonian.
\\In Section \ref{sect2} we define the Environment Action Algebra as a more relevant algebra of the environment and study its properties as mentioned above.

\paragraph{Notations:}
Throughout this paper, we make use of the following notations:
\begin{itemize}
\item Most of the time, we consider a unitary operator $U$ acting on the tensored Hilbert space $\Hcal\otimes\Kcal$, where $\Hcal$ and $\Kcal$ are separable Hilbert space, modeling the system and the environment respectively.
\item $\Bcal(\Hcal)$ is the Banach space of all bounded operators on $\Hcal$.
\item $\Tr_\Kcal:\Lcal_1(\Hcal\otimes\Kcal)\to\Lcal_1(\Hcal)$ (where $\Lcal_1(\Hcal)$ is the space of trace-class operators on $\Hcal$) stands for the partial trace over $\Kcal$, that is, if $\rho$ is a trace-class operator on $\Hcal\otimes\Kcal$ than for any $X\in\Bcal(\Hcal)$,
\begin{equation*}
\Tr\left[\Tr_\Kcal[\rho]X\right]=\Tr\left[\rho\ \left(X\otimes I_\Kcal\right)\right].
\end{equation*}
Given a trace-class operator $\omega\in\Lcal_1(\Kcal)$, the operation of the partial trace is denoted by $\Tr_\omega:\Bcal(\Hcal\otimes\Kcal)\to\Bcal(\Hcal)$.
\item We will also use the Dirac notations: for any elements $e,f\in\Hcal$:
\begin{itemize}
\item $|e>$ is the linear map from $\C$ to $\Hcal$, $\lambda\in\C\mapsto\lambda|e>=\lambda e$. We most of the time identify the linear map $|e>$ with the element $e$ of $\Hcal$.
\item $<e|$ is its dual, that is the linear map from $\Hcal$ to $\C$ such that $g\in\Hcal\mapsto<e|g=\sca{e}{g}$;
\item and consequently $\proj{e}{f}$ stands for the linear operator on $\Hcal$ such that $g\in\Hcal\mapsto \proj{e}{f}g=\sca{f}{g}e$.
\end{itemize}
\item If $\Acal$ is a subset of $\Bcal(\Kcal)$, $\Acal'$ is its commutant, that is the set
\begin{equation*}
\Acal'=\{Y\in\Bcal(\Kcal),[Y,A]=0\text{ for all }A\in\Acal\}.
\end{equation*}
If $\Acal$ is a $*$-stable set, then $\Acal'$ is a von Neumann algebra. In this case, by the Bicommutant Theorem of von Neumann \cite{vN1}, the von Neumann algebra generated by $\Acal$ is the bicommutant $\Acal''$.
\end{itemize}

\label{sectintro}
\section{The Environment Algebra for one-step evolution}

The goal of this section is to prove a decomposition of the environment between a classical part and a quantum part, in a sense we shall define later. This is done by studying the proper operator algebra of the environment, the \emph{Environment Algebra}. In order to motivate the idea behind such a definition we first recall a result of Attal, Deschamps and Pellegrini on classical environment (see \cite{ADP2}).

\begin{theo}
Suppose $\Kcal\approx\C^d$ for some positive integer $d$. Let $U$ be a unitary operator acting on the space $\Hcal\otimes\Kcal$. Then the following assertions are equivalent:
\begin{enumerate}
\item There exists $d$ unitary operators $U_1,...,U_d$ on $\Hcal$ and an orthonormal basis $\psi_i$ of $\Kcal$ such that:
\begin{equation}
U=\sum_{i=1}^d{U_i\otimes\proj{\psi_i}{\psi_i} }.
\label{eqtheoclassicalnoise1}
\end{equation}
\item There exists operators $A,B_1,...,B_d\in\Bcal(\Hcal)$ and an obtuse random variable $X=(X_1,...,X_d)$ on $\C^d$ such that $U$ can be written in some orthonormal basis of $\Kcal$ as
\begin{equation}
U=A\otimes I_\Kcal + \sum_{i=1}^d{B_i\otimes M_{X_i}};
\label{eqtheoclassicalnoise2}
\end{equation}
where $M_{X_i}$ is isomorphic to the multiplication operator by the coordinate random variable $X_i$.
\end{enumerate}
\label{theoclassicalnoise}
\end{theo}

\noindent Thus unitary operators such as in Equation (\ref{eqtheoclassicalnoise1}) display strong classical behavior, as they are linked in a one to one way with some particular class of random variables: the obtuse random variables. As we shall not need them in the following, we do not wish to explain more on the matter. Instead we advise the reader to consult \cite{ADP1} and \cite{A-P2}.

Our first goal in this section is to find a characterization of Equation (\ref{eqtheoclassicalnoise1}) that can be generalized in infinite dimension. This is the role of the Environment Algebra. It is based on the following remark: $U$ is of the form (\ref{eqexcom1}) if and only if it belongs to a von Neumann algebra $\Bcal(\Hcal)\otimes\Acal$ where $\Acal$ is a \emph{commutative} algebra on $\Kcal$. Indeed, if $U$ is of the form (\ref{eqtheoclassicalnoise1}), then define the algebra $\Acal$ as:

\begin{equation*}
\Acal=\left\{\sum_{i=1}^d{f(i)\proj{\psi_i}{\psi_i}},\ f\in L^\infty(\{1,...,p\})\right\}.
\end{equation*}

\noindent Then clearly $U\in\Bcal(\Hcal)\otimes\Acal$. Conversely, assume that there exists a commutative von Neumann algebra $\Acal$ on $\Kcal$ such that $U\in\Bcal(\Hcal)\otimes\Acal$. It is well-known that this algebra takes the form:

\begin{equation}
\Acal=\left\{\sum_{i=1}^m{f(i)P_i},\ f\in L^\infty(\{1,...,m\})\right\},
\end{equation}

\noindent where the $P_i$'s are uniquely defined mutually orthogonal projections that sums to the identity. As $U\in\Bcal(\Hcal)\otimes\Acal$, there exist unitary operators $V_1,...,V_m$ on $\Hcal$ such that:

\begin{equation}
U=\sum_{i=1}^m{V_i\otimes P_i}.
\label{eqexcom1}
\end{equation}

\noindent We will show in Proposition \ref{propenvcom} how this equation can be generalized in any dimension. In general, we will see that the environment is commutative if and only if the Environment Algebra is commutative.

To finish this small introduction on classical environment, let us remark that there are still two additional classical behaviors connected to these unitary operators:
\begin{itemize}
\item When considering several steps of the evolution, one can prove that the resulting evolution is equivalent to a random walk on the unitary group of $\Hcal$ (see \cite{A-D}).
\item For any density matrix $\omega$ on $\Kcal$ and any observable $X\in\Bcal(\Hcal)$, we have
\begin{equation}
\Lcal_{\omega}(X):=\Tr_\omega\left[U^* X\otimes I_\Kcal U\right]=\sum_{i=1}^{d}{\sca{\psi_i}{w\psi_i}U_i^*X U_i}.
\label{eqCPmap}
\end{equation}
The map $\Lcal_\omega$ is thus a \emph{random unitary} completely positive map (CP map), a kind of CP map which is largely study in Quantum Information Theory (see \cite{raey} for a review on that matter).
\\It has been conjecture in \cite{ADP2} by Attal and al. that a unitary operator that gives a random unitary CP map for all density matrices of the environment must have the form (where $\varphi_i$ is another orthonormal basis):
\begin{equation}
U=\sum_{i=1}^d{U_i\otimes\proj{\varphi_i}{\psi_i}}.
\label{eqexcom2}
\end{equation}
\end{itemize}
We will come back on these specific unitary operators in Section \ref{sect2}. Indeed we will show as a corollary of Theorem \ref{theoaction} that they correspond exactly to those that have a commutative \emph{Environment Right-Action Algebra} (see Definition \ref{deactionenvironment}). We will show that $U$ has the form (\ref{eqexcom2}) if and only if this algebra is commutative.

As we already mentioned in the introduction, the Spontaneous Emission, whose evolution is described by the unitary operator $U_{\text{se}}$ of Equation (\ref{eqSpontemiss}), is well-known for being a manifestation of the quantum world. We will see later that this can be interpreted in our framework as the fact that there does not exist a commutative algebra $\Acal$ on $\Kcal$ such that $U_{\text{se}}\in\Bcal(\Hcal)\otimes\Acal$. This can be read directly on the Environment Algebra for this evolution, as we will show that it is the whole algebra $\Bcal(\Kcal)$. Consequently both situations can be integrated in the framework of the Environment Algebra. 
\\The general case stands in between the classical case and the quantum case, as it can occur that some subpart only of the environment displays classical behavior. The second goal of this section is to define what is a classical part of the environment. Then, in the general case, we want to be able to identify all the classical parts of the environment. Once again this will be done using the Environment Algebra.

In Subsection \ref{sect11} below we define the \emph{Environment Algebra} $\Acal(U)$ as a relevant subalgebra of $\Bcal(\Kcal)$. We then use this algebra to properly define a commutative environment, and to characterize unitary operators with commutative environment. In Subsection \ref{sect12} we define the classical parts of the environment, namely the \emph{Commutative Subspaces of the Environment}, and we prove the existence of a maximal commutative subspace, that contains all the others. This provides us with a decomposition of the environment between a classical and a quantum part. In Subsection \ref{sect13} we study such a decomposition on one example derived from a typical Hamiltonian.

\label{sect1}
\subsection{Definition, first characterization and first examples}

Let $U$ be a unitary operator on $\Hcal\otimes\Kcal$. We will need the following notation: for $f,g\in\Hcal$, we define:

\begin{equation}
U(f,g)=\Tr_{\proj{g}{f}}[U],\ U^*(f,g)=\Tr_{\proj{g}{f}}[U^*].
\label{eqU(f,g)}
\end{equation}

Those operators can be seen as pictures of $U$ taken from $\Kcal$ but with different angles.

\begin{de}
Let $U$ be a unitary operator on $\Hcal\otimes\Kcal$. We call the \emph{Environment Algebra} the von Neumann algebra $\Acal(U)$ generated by the $U(f,g)$, that is
\begin{equation}
\Acal(U)=\left\lbrace U(f,g),\ U^*(f,g);\quad f,g\in\Hcal\right\rbrace''.
\label{eqdeenvironmentalg}
\end{equation}
\label{deenvironment}
\end{de}

\noindent The point with this definition is that it fits with the following characterization. 

\begin{prop}
Let $U$ be a unitary operator on $\Hcal\otimes\Kcal$. Then $\Acal(U)$ is the smallest von Neumann subalgebra of $\Bcal(\Kcal)$ such that $U,U^*\in\Bcal(\Hcal)\otimes\Acal(U)$, i.e. if $\Acal$ is another von Neumann algebra such that $U,U^*\in\Bcal(\Hcal)\otimes\Acal$, then $\Acal(U)\subset\Acal$. Furthermore, its commutant is given by
\begin{equation}
\Acal(U)'=\left\lbrace Y\in\Bcal(\Kcal),\quad [I_\Hcal\otimes Y,U]=[I_\Hcal\otimes Y,U^*]=0\right\rbrace.
\label{eqpropenvironment}
\end{equation}
\label{propenvironment}
\end{prop}

\begin{proof}
First, if $(e_i)$ is an orthonormal basis of $\Hcal$, then $U$ has the matrix decomposition (see \cite{Attb}):
\begin{equation*}
U=\sum_{i,j}{\proj{e_i}{e_j}\otimes U(e_i,e_j)},
\end{equation*}
where the sum is strongly convergent if $\Hcal$ is infinite dimensional. As the same decomposition holds for $U^*$, we obtain that $U,U^*\in\Bcal(\Hcal)\otimes\Acal(U)$. Now for all $Y\in\Bcal(\Kcal)$,
\begin{align*}
[I_\Hcal\otimes Y,U]=0
& \Leftrightarrow \Tr_{\proj{g}{f}}\left[[I_\Hcal\otimes Y,U]\right]=0 \text{ for all }f,g\in\Hcal \\
& \Leftrightarrow \left[Y,\Tr_{\proj{g}{f}}\left[U\right]\right]=0 \text{ for all }f,g\in\Hcal \\
& \Leftrightarrow [Y,U(f,g)]=0 \text{ for all }f,g\in\Hcal.
\end{align*}
Similarly $[I_\Hcal\otimes Y,U^*]=0$ if and only if $[Y,U^*(f,g)]=0$ for all $f,g\in\Hcal$ which proves Equality (\ref{eqpropenvironment}).
\\Now suppose that $U,U^*\in\Bcal(\Hcal)\otimes\Acal$ for some von Neumann subalgebra $\Acal$ of $\Bcal(\Kcal)$. Then for all $Y\in\Acal'$, $[I_\Hcal\otimes Y,U]=[I_\Hcal\otimes Y,U^*]=0$ so that by the previous equality $Y\in\Acal(U)'$. Consequently $\Acal'\subset\Acal(U)'$ and then $\Acal(U)\subset\Acal$ by the Bicommutant Theorem.
\end{proof}

We now show how those definitions apply to the two examples mentioned in the introduction: the general situation of a commutative algebra and the particular case of the spontaneous emission.

\paragraph{The case of a Commutative Environment:}
In this paragraph we characterize the unitary operators with commutative Environment Algebra. The expression in the general case is of the same form as Equation (\ref{eqexcom1}), where the sum has to be replaced by an integral over a spectral measure. Note that in general, if a von Neumann algebra $\Acal$ on $\Kcal$ is commutative, then there exist a measured space $(\Omega,\Fcal)$ and a spectral measure $\xi$ on $(\Omega,\Fcal)$ with values in the orthogonal projections of $\Kcal$ such that:

\begin{equation}
\Acal=\left\{\int_\Omega{f(\omega)\xi(d\omega)},\quad f\in L^\infty(\Omega,\Fcal)\right\}.
\label{eqalgcom}
\end{equation}

\noindent As $U\in\Bcal(\Hcal)\otimes\Acal(U)$, the goal is to understand the nature of this last algebra. The next proposition shows that it consists of operators of the form (\ref{eqalgcom}) where the function $f$ has to be replaced by a measurable family of bounded operators on $\Hcal$.

\begin{prop}
Let $U$ be a unitary operator on $\Hcal\otimes\Kcal$. Then $\Acal(U)$ is commutative if and only if there exist:
\begin{itemize}
\item a measured space $(\Omega,\Fcal)$,
\item a spectral measure $\xi$ on $(\Omega,\Fcal)$ with values in the orthogonal projections of $\Kcal$,
\item a $\Fcal$-measurable family of $\xi$-almost surely unitary operators
\\$(V(\omega))_{\omega\in\Omega}$ on $\Hcal$ such that:
\begin{equation}
U= \int_{\Omega}{V(\omega)\otimes \xi(d\omega)}.
\label{eqpropenvcom}
\end{equation}
\end{itemize}
\label{propenvcom}
\end{prop}

If $\Kcal$ is finite dimensional, then Equation (\ref{eqpropenvcom}) takes the simpler form of Equation (\ref{eqexcom1}).

\begin{proof}
If $\Acal(U)$ is commutative then it is of the same form as in Equation (\ref{eqalgcom}). 
\\The first step of the proof is to construct the kind of operators that appear in Equation (\ref{eqpropenvcom}). Let $B(\Omega)$ denote the set of $\xi$-almost bounded families of bounded operators on $\Hcal$ indexed by $\Omega$. Thus $A\in B(\Omega)$ is a random variable on $(\Omega,\Fcal)$ with value in $\Bcal(\Hcal)$ and such that there exists a constant $C>0$ with $\norm{A}<C$, $\xi$-almost surely.
\\We want to integrate with respect to $\xi$ the elements of $B(\Omega)$. For $\varphi,\psi\in\Kcal$, define the complex measure $\nu_{\varphi,\psi}$ on $(\Omega,\Fcal)$ as
\begin{equation*}
\nu_{\varphi,\psi}:E\in\Fcal\mapsto \sca{\varphi}{\xi(E)\psi}.
\end{equation*}
Then for all $A\in B(\Omega)$ and all $f,g\in\Hcal$, we have
\begin{align*}
\left|\int_\Omega{\sca{f}{A(\omega)g}\nu_{\varphi,\psi}(d\omega)}\right|
& \leq \int_\Omega{\left|\sca{f}{A(\omega)g}\right|\left|\sca{\varphi}{\xi(d\omega)\psi}\right|} \\
& \leq \norm{A}\norm{f}\norm{g}\norm{\varphi}\norm{\psi}.
\end{align*}
By Riesz representation Theorem \cite{R-S} this defines a bounded operator on $\Hcal\otimes\Kcal$ that we write $\int_\Omega{A\otimes d\xi}$ and such that
\begin{equation}
\sca{f\otimes\varphi}{\left(\int_\Omega{A\otimes d\xi}\right)g\otimes\psi}=\int_\Omega{\sca{f}{A(\omega)g}\nu_{\varphi,\psi}(d\omega)}.
\label{eqproofenvcom}
\end{equation}
We call $\Bcal(\xi)$ the set of operators defined by Equation (\ref{eqproofenvcom}). Note that elements of $\Bcal(\xi)$ can also be defined as weak limits of operators of the form
\begin{equation}
\sum_{j\in J}{A_j\otimes\xi(E_j)},
\label{eqproofenvcom2}
\end{equation}
where $J$ is a finite set, $(A_j)_{j\in J}$ is a family in $\Bcal(\Hcal)$ and $(E_j)_{j\in J}$ is a measurable partition of $(\Omega,\Fcal)$.
From this remark and Equation (\ref{eqproofenvcom}), it is an easy exercise to check that $\Bcal(\xi)$ is an algebra and that it is closed for the weak topology, so that it is a von Neumann algebra on $\Hcal\otimes\Kcal$. Besides, if $\int_\Omega{A\otimes d\xi}=\int_\Omega{B\otimes d\xi}$ for $A,B\in\Bcal(\xi)$, then $A=B$, $\xi$-almost surely.
\\Note that if there exists $V\in B(\Omega)$ such that $U=\int_\Omega{V\otimes d\xi}$, then, from $UU^*=U^*U=I_{\Hcal\otimes\Kcal}$, we get $V{V}^*={V}^*V=I_{\Hcal}$, $\xi$-almost surely, which shows that the $V(\omega)$ are $\xi$-almost surely unitary operators on $\Hcal$. Consequently, as $U$ belongs to $\Bcal(\Hcal)\otimes\Acal(U)$, in order to complete the proof we only have to show that
\begin{equation*}
\Bcal(\xi)=\Bcal(\Hcal)\otimes\Acal(U).
\end{equation*}
We first prove that $\Bcal(\xi)\subset\Bcal(\Hcal)\otimes\Acal(U)$. By the Bicommutant Theorem, this is equivalent to $I_\Hcal\otimes\Acal(U)'\subset\Bcal(\xi)'$. Let $Y\in\Acal(U)'$. It is enough to prove that $I_\Hcal\otimes Y$ commutes with elements of the form (\ref{eqproofenvcom2}), which is straightforward.
\\To prove the other inclusion, take $X\in\Bcal(\Hcal)$ and $A\in\Acal(U)$. Then there exists a bounded measurable function $f$ on $(\Omega,\Fcal)$ such that $A=\int_\Omega{fd\xi}$. Then
\begin{equation*}
X\otimes A=\int_\Omega{(f(\omega)X)\ \xi(d\omega)}\in\Bcal(\xi).
\end{equation*}
This concludes the proof as the von Neumann algebra $\Bcal(\Hcal)\otimes\Acal$ is the weak closure of operators of the form $X\otimes A$. 
\end{proof}

\paragraph{An example of purely quantum environment: the Spontaneous Emission:}
Consider the unitary operator $U_{\text{se}}$ given by Equation (\ref{eqSpontemiss}). We suppose that $\theta\notin 2\pi\Zbb$ so that $U_{\text{se}}$ is not the identity operator. By common knowledge it is a purely quantum evolution, so the corresponding environment algebra should not be commutative. More particularly, as $\Kcal=\C^2$, it should be the whole algebra.
\\In order to check this, take an operator $Y\in\Acal(U_{\text{se}})'$. Let $(e_0,e_1)$ be the canonical basis of $\C^2$ and identify $Y$ with the matrix $\begin{pmatrix} a & b \\ c & d \end{pmatrix}$, $a,b,c,d\in\C$. The operator $I_\Hcal\otimes Y$ is identified with the $4\times4$-matrix $\begin{pmatrix} aI_\Hcal & bI_\Hcal \\ cI_\Hcal & dI_\Hcal \end{pmatrix}$. Computing $[I_\Hcal\otimes Y,U_{\text{se}}]$ leads to:
\begin{align*}
& [I_\Hcal\otimes Y,U_{\text{se}}] = \\
& \begin{pmatrix}
0 & b\sin\theta & b(\cos\theta-1) & 0 \\
c\sin\theta & 0 & (d-a)\sin\theta & b(1-\cos\theta) \\
c(1-\cos\theta) & (d-a)\sin\theta & 0 & -b\sin\theta \\
0 & c(\cos\theta-1) & -c\sin\theta & 0
\end{pmatrix}.
\end{align*}
Hence the condition $[I_\Hcal\otimes Y,U_{\text{se}}]=0$ implies that $b=c=0$ and $a=d$, so that $Y=aI_\Kcal$. Thus $\Acal(U_{\text{se}})'=\C I_\Kcal$ and consequently
\begin{equation*}
\Acal(U_{\text{se}})=\Bcal(\C^2).
\end{equation*}

\label{sect11}
\subsection{Classical and Quantum parts of the Environment}

We now focus on our main topic. We have already given some examples of the algebra $\Acal(U)$ and emphasized on the particular case where it is commutative. In this subsection we state our main result for a general one-step unitary evolution: the environment can be splited into the sum of a classical and a quantum part. First using the algebra $\Acal(U)$ it is now possible to properly define what we call a classical part, or commutative part, of the environment. For instance, consider the following unitary operator $U_{\text{ex}}$ on $\C^2\otimes\C^4$ ($\alpha,\beta\in\R$, $\alpha\ne\beta$), written in the canonical orthonormal basis $(e_1,e_2,e_3,e_4)$ of $\Kcal=\C^4$:

\begin{equation}
U_{\text{ex}}=\begin{pmatrix}
\cos\alpha & -\sin\alpha & 0 & 0 & 0 & 0 & 0 & 0 \\
\sin\alpha & \cos\alpha & 0 & 0 & 0 & 0 & 0 & 0 \\
0 & 0 & \cos\beta & -\sin\beta & 0 & 0 & 0 & 0 \\
0 & 0 & \sin\beta & \cos\beta & 0 & 0 & 0 & 0 \\
0 & 0 & 0 & 0 & 1 & 0 & 0 & 0 \\
0 & 0 & 0 & 0 & 0 & \cos\theta & -\sin\theta & 0 \\
0 & 0 & 0 & 0 & 0 & \sin\theta & \cos\theta & 0 \\
0 & 0 & 0 & 0 & 0 & 0 & 0 & 1 
\end{pmatrix}.
\label{eqexunitaire}
\end{equation}

\noindent Clearly the environment is the sum of two parts $\Kcal_c=\C e_1\oplus \C e_2$ and $\Kcal_q=\C e_3\oplus \C e_4$, such that $\Hcal\otimes\Kcal_{c,q}$ are stable by $U$ and:

\begin{itemize}
\item The restriction $U_c$ of $U$ on $\Hcal\otimes\Kcal_c$ has a commutative environment. Indeed:
\begin{align*}
U_c & =\begin{pmatrix}
\cos\alpha & -\sin\alpha & 0 & 0 \\
\sin\alpha & \cos\alpha & 0 & 0  \\
0 & 0 & \cos\beta & -\sin\beta \\
0 & 0 & \sin\beta & \cos\beta 
\end{pmatrix} \\
& =\begin{pmatrix}\cos\alpha & -\sin\alpha \\ \sin\alpha & \cos\alpha \end{pmatrix}\otimes\proj{e_1}{e_1}
+\begin{pmatrix}\cos\beta & -\sin\beta \\ \sin\beta & \cos\beta \end{pmatrix}\otimes\proj{e_2}{e_2},
\end{align*}
so that $\Acal(U_c)=\C\proj{e_1}{e_1}+\C\proj{e_2}{e_2}=\left\{\begin{pmatrix}a&0\\0&b\end{pmatrix},\ a,b\in\C\right\}$ as $\alpha\ne\beta$.
\item The restriction $U_q$ of $U$ on $\Hcal\otimes\Kcal_q$ has a quantum environment, as $U_q=U_{\text{se}}$ so that $\Acal(U_q)=\Bcal(\Kcal_q)$.
\end{itemize}

\noindent Consequently in this example $\Kcal_c$ is a classical subspace of the environment, which is characterized by the fact that $\Hcal\otimes\Kcal_c$ reduces both $U$ and $U^*$ and that $\Acal(U_c)$ is commutative. Furthermore the other non-trivial parts of the environment that could be considered as classical, that is $\C e_1$ and $\C e_2$, are subspaces of $\Kcal_c$.
\\This leads us to the following definition.

\begin{de}
Let $U$ be a unitary operator on $\Hcal\otimes\Kcal$ and let $\tilde{\Kcal}$ be a subspace of $\Kcal$. We say that $\tilde{\Kcal}$ is a \emph{Commutative Subspace of the Environment} if $\tilde{\Kcal}\ne\{0\}$ and:
\begin{itemize}
\item $\Hcal\otimes\tilde{\Kcal}$ and $\Hcal\otimes\tilde{\Kcal}^\perp$ are stable by $U$,
\item $\Acal(\tilde{U})$ is commutative, where $\tilde{U}$ is the restriction of $U$ to $\tilde{\Kcal}$.
\end{itemize}
\label{ClassicalQuantumpart}
\end{de}

Our main result, Theorem \ref{theodecompenv}, states that there exists a maximal commutative subspace $\Kcal_c$ of the environment, in the sense that all commutative subspaces of the environment are also subspaces of $\Kcal_c$. The main ingredient of the proof is the following proposition.

\begin{prop}
Let $\Acal$ be a von Neumann subalgebra of $\Bcal(\Kcal)$. Then there exists a unique projection $P_c\in\Acal'$, possibly null, such that:
\begin{enumerate}
\item $P_c\ \Acal\ P_c$ is commutative;
\item if $P\in\Acal'$ is an orthogonal projection such that $P \Acal\ P$ is commutative, then $P\leq P_c$.
\end{enumerate}
\label{propdecompaction}
\end{prop}

\noindent This proposition is not a new result on von Neumann algebras, however we could not find it stated in this particular form in the literature. It is true for any von Neumann algebra and entirely relies on the fact that the supremum of a subclass of orthogonal projections in a von Neumann algebra still belongs to this algebra.

\begin{proof}
Define the set $\Pcal_c$ of orthogonal projections in $\Acal'$ such that $P\in\Pcal_c$ iff $P\Acal$ is commutative. Take $P_c=\sup\Pcal_c$ the supremum over $\Pcal_c$. It is again in $\Pcal_c$ and it is easy to see that it fulfills the proposition.
\end{proof}

\noindent Our result is a direct corollary of this proposition.

\begin{theo}
The environment Hilbert space $\Kcal$ is the orthogonal direct sum of two subspaces $\Kcal_c$ and $\Kcal_q$, such that either $\Kcal_c=\{0\}$ or:
\begin{itemize}
\item$\Kcal_c$ is a commutative subspace of the environment.
\item If $\tilde{\Kcal}$ is any commutative subspace of the environment then $\tilde{\Kcal}$ is a subspace of $\Kcal_c$.
\item The restriction of $U$ to $\Hcal\otimes\Kcal_q$ does not have any commutative subspace.
\end{itemize}
\label{theodecompenv}
\end{theo}

\begin{proof}
We take $\Kcal_c=P_c \Kcal$ and $\Kcal_q=(I_\Kcal-P_c)\Kcal=\Kcal_c^\perp$, where $P_c$ is defined by applying Proposition \ref{propdecompaction} to $\Acal(U)$. 
\\Suppose that $P_c\ne0$. First we check that $\Kcal_c$ is a commutative subspace of the environment. By definition $P_c\in\Acal(U)'$ so that $I_\Hcal\otimes P_c\in\left(\Bcal(\Hcal)\otimes\Acal(U)\right)'$. As $U$ and $U^*$ are elements of $\Bcal(\Hcal)\otimes\Acal(U)$, this shows that $\Hcal\otimes\Kcal_c$ and $\Hcal\otimes\Kcal_c^\perp$ are stable by $U$. Denote by $U_c$ the restriction of $U$ to $\Hcal\otimes\Kcal_c$. Then $\Acal(U_c)=P_c\Acal(U)P_c$, which is commutative by definition of $P_c$.
\\Let $\tilde{\Kcal}$ be a commutative subspace of $U$ and denote by $P$ the orthogonal projection on $\tilde{\Kcal}$ and by $\tilde{U}$ the restriction of $U$ to $\tilde{\Kcal}$. By definition $P\in\Acal(U)'$ and $\Acal(\tilde{U})=P\Acal(U) P$ is commutative, so $P\leq P_c$ by Proposition \ref{propdecompaction}. Consequently $\tilde{\Kcal}$ is a subspace of $\Kcal_c$.
\\Now denote by $U_q$ the restriction of $U$ to $\Hcal\otimes\Kcal_q$. By contradiction, let $\Kcal'$ be a commutative subspace of $U_q$. By the same argument as before $\Kcal'$ is also a commutative subspace of $U_q$. Consequently $\Kcal'$ is a subspace of $\Kcal_c$, so that $\tilde{\Kcal}=\{0\}$, which is a contradiction.
\end{proof}

\label{sect12}
\subsection{Typical Hamiltonian: the Dipole Hamiltonian}

In this section, $\Hcal$ is a $N$-dimensional Hilbert space and $\Kcal$ is a $(d+1)$-dimensional Hilbert space, with $N,d$ two positive integers. If $U$ is a unitary operator on $\Hcal\otimes\Kcal$ it is always possible to write it $U=\exp{(-itH)}$ for some selfadjoint operator $H$ (which is not unique). Then it is interesting to study the algebra $\Acal(H)$ instead. It is also a first step towards understanding the continuous case, as the Hamiltonian describes the instantaneous evolution of the system and its environment (see \cite{A-P1}). We focus on the particular case of a dipole Hamiltonian.

Let $(e_i)_{i=0,...,d}$ be an orthonormal basis of $\Kcal$, starting the index at $0$. The vector $e_0$ will play a specific role in what follows. We write $\Kcal'=\C e_0^\perp$ its orthogonal subspace. We will often identify operators on $\Kcal$ with $(d+1)$-square matrices, and operators on $\Kcal'$ with $d$-square matrices, using this basis. We also write $a_j^i=\proj{e_j}{e_i}$ the elementary matrices. Writing $V_j^i=\Tr_\Hcal\left[V I_\Hcal\otimes\proj{e_i}{e_j}\right]$, any element $V\in\Bcal(\Hcal\otimes\Kcal)$ can then be written as a block matrix

\begin{equation*}
V=\sum_{i=0}^{d+1}{V_j^i\otimes a_j^i}=\begin{pmatrix}
V_0^0 & V_0^1 & \cdots & V_0^d \\
V_1^0 & V_1^1 & \cdots & V_1^d \\
\vdots & \vdots & & \vdots\\
V_d^0 & V_d^1 & \cdots & V_d^d
\end{pmatrix}.
\end{equation*}

\noindent In the same way elements of $\Bcal(\Hcal)^{\oplus d}$ can be seen as columns whose components are operators on $\Hcal$.
\\Similarly elements of the dual $\left(\Bcal(\Hcal)^{\oplus d}\right)^*$ of $\Bcal(\Hcal)^{\oplus d}$ can be seen as rows whose components are operators on $\Hcal$. If $L\in\Bcal(\Hcal)^{\oplus d}$, we write $L^T$ and $L^*$ the following elements of $\left(\Bcal(\Hcal)^{\oplus d}\right)^*$:

\begin{equation*}
L=\begin{pmatrix}
L_1 \\ \vdots \\ L_d
\end{pmatrix},\quad L^T=\begin{pmatrix}
L_1 & \cdots & L_d
\end{pmatrix},\quad L^*=\begin{pmatrix}
L_1^* & \cdots & L_d^*
\end{pmatrix}.
\end{equation*}

\noindent Elements $M=(m_{i,j})_{1\leq i,j\leq d}$ of $\Bcal(\Kcal')$ act on $\Bcal(\Hcal)^{\oplus d}$ in the following way:

\begin{equation*}
M: L \mapsto \begin{pmatrix}
m_{1,1}I_\Hcal & \cdots & m_{1,d}I_\Hcal \\
\vdots & & \vdots\\
m_{d,1}I_\Hcal & \cdots & m_{d,d}I_\Hcal
\end{pmatrix}\begin{pmatrix}
L_1 \\ \vdots \\ L_d
\end{pmatrix}.  
\end{equation*}

\noindent We write this action $M\star L$, with dual action on the dual $\left(\Bcal(\Hcal)^{\oplus d}\right)^*$ of $\Bcal(\Hcal)^{\oplus d}$:

\begin{equation*}
L^T\star M=\begin{pmatrix}
L_1 & \cdots & L_d
\end{pmatrix}\begin{pmatrix}
m_{1,1}I_\Hcal & \cdots & m_{1,d}I_\Hcal \\
\vdots & & \vdots\\
m_{d,1}I_\Hcal & \cdots & m_{d,d}I_\Hcal
\end{pmatrix}.
\end{equation*}

Our result makes use of the following lemma, which is straightforward by a simple computation on block-matrices.

\begin{lem}
Let $W$ be a unitary operator on $\Kcal'$ and write $\overset{\longrightarrow}{W}=a_0^0\oplus W$. Thus $\overset{\longrightarrow}{W}$ is the unitary operator on $\Kcal$ that acts as identity on $\C e_0$ and as $W$ on $\Kcal'$. Let $V\in\Bcal(\Hcal\otimes\Kcal)$. We write $L_0=(V_0^1\cdots V_0^d)$, $L^0=(V_1^0 \cdots V_d^0)^T$ and $\mathbb{L}=(V_j^i)_{1\leq i,j\leq d}$, so that:
\begin{equation*}
V=\begin{pmatrix}
V^0_0 & L_0 \\
L^0 & \mathbb{L}
\end{pmatrix}.
\end{equation*}
\noindent Then
\begin{equation*}
\left(I_\Hcal\otimes\overset{\longrightarrow}{W}^*\right)\ V\ \left(I_\Hcal\otimes\overset{\longrightarrow}{W}\right) =\begin{pmatrix}
V^0_0 & L_0\star W \\
W^*\star L^0 & \left(I_\Hcal\otimes W^*\right)\mathbb{L}\left(I_\Hcal\otimes W\right)
\end{pmatrix}.
\end{equation*}
\label{lemunittransfom}
\end{lem}

We focus on the typical dipole Hamiltonian usually considered in the weak coupling limit or van Hove limit \cite{Dav1}:

\begin{equation*}
H=H_\Scal\otimes I_\Kcal + I_\Hcal\otimes H_\Ecal + \sum_{i=1}^{d}[V_i\otimes a_i^0 + V_i^*\otimes a_0^i],
\end{equation*}

\noindent where $H_\Scal,V_1,...,V_d\in\Bcal(\Hcal)$, $H_\Ecal\in\Bcal(\Kcal)$ and $H_\Scal,H_\Ecal$ are selfadjoint operators. To simplify the notations, we write $V=\begin{pmatrix} V_1 & \cdots & V_d\end{pmatrix}^T$, so that we have in the orthonormal basis $(e_i)_{0\leq i\leq d}$:

\begin{equation}
H=H_\Scal\otimes I_\Kcal + I\otimes H_\Ecal +\begin{pmatrix}
0 & V^* \\
V & 0_{\Kcal'}
\end{pmatrix}.
\label{eqdipolham}
\end{equation}

\noindent First remark that $\Acal(H)'$ is the commutant of the set

\begin{equation*}
\left\{H_\Ecal,\ \sum_{i=1}^{d}[\sca{f}{V_ig}\otimes a_i^0 + \sca{f}{V_i^*g}\otimes a_0^i],\ f,g\in\Hcal\right\}.
\end{equation*}

\noindent To make it simpler, we will suppose that $H_\Ecal=0$ (that is, we switch to the interaction picture of time evolution), so that

\begin{equation}
\Acal(H)'=\left\{Y\in\Bcal(\Kcal),\ \left[I_\Hcal\otimes Y,\begin{pmatrix}
0 & V^* \\
V & 0_{\Kcal'}
\end{pmatrix}\right]=0\right\}'.
\label{eqcomham}
\end{equation}

\noindent Here is our result.

\begin{theo}
Suppose $\Kcal=\C^{d+1}$, with $d\geq1$, and
\begin{equation*}
H=H_\Scal\otimes I_\Kcal+\begin{pmatrix}
0 & V^* \\
V & 0_{\Kcal'}
\end{pmatrix},
\end{equation*}
for some $V=\begin{pmatrix}V_1 & \cdots & V_d\end{pmatrix}^T\in\Bcal(\Hcal)^{\oplus d}$. Let $m$ be the dimension of the subspace of $\Bcal(\Hcal)$ generated by the $V_i$'s. Then we are in one of the following situations:
\begin{itemize}
\item either there exist $\theta\in\R$, $a_1,...,a_d\in\C$ possibly null, such that
\begin{equation}
H=H_\Scal\otimes I_\Kcal+e^{-i\theta/2}V_1\otimes\begin{pmatrix}
0 & e^{-i\theta/2}\overline{a_1} & \cdots & e^{-i\theta/2}\overline{a_d}\\
e^{i\theta/2}a_1 & 0 & \cdots  & 0 \\
\vdots & \vdots &  & \vdots \\
e^{i\theta/2}a_d & 0 & \cdots  & 0\\
\end{pmatrix}.
\label{eqtheodipolham}
\end{equation}
In this case
\begin{equation*}
\Acal(H)=\text{alg}\left\{\begin{pmatrix}
0 & e^{-i\theta/2}\overline{a_1} & \cdots & e^{-i\theta/2}\overline{a_d}\\
e^{i\theta/2}a_1 & 0 & \cdots  & 0 \\
\vdots & \vdots &  & \vdots \\
e^{i\theta/2}a_d & 0 & \cdots  & 0\\
\end{pmatrix}\right\},
\end{equation*}
and consequently it is commutative.
\item either $\Kcal$ is the orthogonal sum of two subspaces $\Kcal_1\cong\C^{m+1}$ and $\Kcal_2\cong\C^{d-m}$, such that $\Hcal\otimes\Kcal_1$ and $\Hcal\otimes\Kcal_2$ are stable under $H$ and $\Acal(H)$ can be decomposed as
\begin{equation}
\Acal(H)=\Bcal(\Kcal_1)\oplus\C I_{\Kcal_2}.
\label{eqtheoalgdipolham}
\end{equation}
\end{itemize}
\label{theoalgdipolham}
\end{theo}

\noindent Before giving the proof of Theorem \ref{theoalgdipolham}, we introduce the following lemma which allows to reduce the problem to the minimal number of non-zero $V_k$ and to assume their freeness.

\begin{lem}
Suppose $\Kcal=\C^{d+1}$, with $d\geq1$, and
\begin{equation}
H=H_\Scal\otimes I_\Kcal+\begin{pmatrix}
0 & V^* \\
V & 0_{\Kcal'}
\end{pmatrix},
\label{eqlemdipolham1}
\end{equation}
for some $V=\begin{pmatrix}V_1 & \cdots & V_d\end{pmatrix}^T\in\Bcal(\Hcal)^{\oplus d}$. Let $m$ be the dimension of the subspace of $\Bcal(\Hcal)$ generated by the $V_i$'s.
\\Then there exists an orthonormal basis $(e'_i)_{1\leq i \leq d}$ of $\Kcal'$ such that in the new basis $(e_0,e'_1,...,e'_d)$, $H$ has the form:
\begin{equation}
H=H_\Scal\otimes I_\Kcal+\begin{pmatrix}
0 & (V')^* \\
V' & 0_{\Kcal'}
\end{pmatrix},
\label{eqlemdipolham2}
\end{equation}
\noindent where $V'=\begin{pmatrix} V'_1 & \cdots V'_m & 0 & \cdots & 0 \end{pmatrix}^T\in\Bcal(\Hcal)^{\oplus d}$.
\\Furthermore, the $V'_i$'s are linearly independent, and $d-m$ is the maximal number of components that can be canceled.
\label{lemdipolham}
\end{lem}

\begin{proof}[Proof of Lemma \ref{lemdipolham}]
If $m=0$, then $V_1=\cdots=V_d=0$ so that the result is trivial. Suppose that $m\geq1$.
\\We identify $V_k$ with the matrix $(v^{i,j}_k)_{1\leq i,j\leq N}$ in some orthonormal basis of $\Hcal$. Our goal is to find a unitary operator $W$ on $\Kcal'$ such that only the first $m$ components of $W^*\star V$ are non-zero.
\\Introduce the vector $\vec{v}_{i,j}=(v^{i,j}_1,\cdots,v^{i,j}_d)$ of $\C^d$. We will use the two following spaces:
\begin{align*}
& \Vcal_1=\text{span}\{V_1,...,V_d\}\subset\Bcal(\Hcal),\\
& \Vcal_2=\text{span}\{\vec{v}_{i,j},\ i,j=1,...,N\}\subset\C^d.
\end{align*}
If $W$ is a unitary operator on $\Kcal'$, we write similarly
\begin{equation*}
V'=\begin{pmatrix}
V'_1\\ \vdots\\V'_d
\end{pmatrix}=W^*\star V,
\end{equation*}
where $V'_k=((v')^{i,j}_k)_{1\leq i,j\leq N}$. If we write $\vec{v'}_{i,j}=((v')^{i,j}_1,\cdots,(v')^{i,j}_d)$, then for all $i,j=1,...,N$
\begin{equation*}
\vec{v'}_{i,j}=W^*\vec{v}_{i,j}.
\end{equation*}
Consequently, canceling the last $d-k$ components of $V$, for $k=1,...,d$, is equivalent to finding $W$ such that for all $i,j=1,...,N$:
\begin{align*}
&(v')^{i,j}_l\ne 0\text{ for all }l=1,...,k,\\
&(v')^{i,j}_l= 0\text{ for all }l=k+1,...,d.
\end{align*}
The maximal number of components that we can cancel is thus $d-\dim \Vcal_2$. Thus we have to prove that $\dim \Vcal_2=\dim\Vcal_1=m$ to finish the proof. Let $A$ be the $N^2\times d$-matrix whose rows are the $\vec{v}_{i,j}$. For $\lambda_1,...,\lambda_d\in\C$, we have
\begin{equation*}
\begin{pmatrix}
\lambda_1 I_\Hcal & \cdots & \lambda_d I_\Hcal
\end{pmatrix} V=0\quad \text{iff}\quad A\begin{pmatrix}
\lambda_1 \\ \vdots \\ \lambda_d
\end{pmatrix}=0.
\end{equation*}
Consequently, the rank of $A$ is equal to $\dim \Vcal_1=m$. However by definition it is equal to the dimension of the subspace spanned by its rows, that is $\dim \Vcal_2$. Hence the result.
\end{proof}

\noindent Now we can prove Theorem \ref{theoalgdipolham}.

\begin{proof}[Proof of Theorem \ref{theoalgdipolham}]
If $m=0$, the result is trivial. Suppose that $m\geq1$.
\\Using Lemma \ref{lemdipolham} we assume that only the $m$ first $V_k$'s are non-zero, and that they are independent. Our method in order to study $\Acal(H)$ is the following: take $Y\in\Acal(H)'$ and solve in $Y$ the equation $[I_\Hcal\otimes Y,H]=0$. With our assumptions on the $V_k$'s, it is clear that $\Kcal'$ is the orthogonal direct sum of two subspaces $\Kcal_1$ and $\Kcal_2$, of dimension $m$ and $d-m$ respectively, such that
\begin{align*}
& H=H_1\oplus 0_{\Hcal\otimes\Kcal_2}; \\
& \Acal(H)= \Acal(H_1)\oplus \C I_{\Kcal_2};
\end{align*}
where $H_1$ is the selfadjoint operator induced by $H$ on $\Hcal\otimes\left(\C e_0\oplus\Kcal_1\right)$. Thus we can restrict the study to $H_1$ only.
\\Take $Y\in\Acal(H_1)'$. As $\Acal(H_1)'$ is a $*$-algebra, we can assume without loss of generality that $Y$ is selfadjoint. We identify $Y$ with the matrix $(y_{ij})_{0\leq i,j\leq m}$ using the orthonormal basis $(e_i)_{0\leq i\leq m}$ of $\Kcal_1$. Now as $y_{00}I_{\Kcal}\in\Acal(H_1)'$, taking $Y-y_{00}I_{\Kcal}$ we can furthermore assume that $y_{00}=0$. Then, solving $[I_\Hcal\otimes Y,H_1]=0$, we obtain in particular the following conditions:
\begin{align}
& \forall i=1,...,m,\quad \sum_{j=1}^m{y_{ij}V_j}=0; \label{eqcond1}\\
& \forall i,j=1,...,m,\quad y_{i0}V_j^*=y_{0j}V_i. \label{eqcond2}
\end{align}
As the $V_i's$ are non-zero, condition (\ref{eqcond2}) implies that the complex $y_{i0}$ are all zero or all non-zero for $i=1,...,m$.
\begin{itemize}
\item If they are non-zero, the same condition implies that the dimension of the space $\text{span}\left\lbrace V_1,...,V_m\right\rbrace$ is one, so that $m=1$. In this case, writing $\frac{y_{01}}{y_{10}}=a\in\C$, condition (\ref{eqcond2}) is just $V_1^*=aV_1$ so that $V_1=|a|^2V_1$ and consequently $|a|=1$. We write $a=e^{i\theta}$ with $\theta\in\R$ and we obtain (recall that we are in the case $m=1$):
\begin{equation*}
H_1=H_\Scal\otimes I_\Kcal+e^{-i\theta/2}V_1\otimes\begin{pmatrix}
0 & e^{-i\theta/2} \\
e^{i\theta/2} & 0
\end{pmatrix}.
\end{equation*}
Furthermore it is straightforward that for any unitary operator $W$ on $\Kcal'$,
\begin{equation*}
W^*\star\begin{pmatrix}
V_1 \\ 0 \\ \vdots \\ 0
\end{pmatrix}=\begin{pmatrix}
\sca{e_1}{W^*e_1} V_1 \\ \vdots \\ \sca{e_d}{W^*e_1}V_1
\end{pmatrix}.
\end{equation*}
Writing $a_i=\sca{e_i}{W^*e_1}$ for all $i=1,...,d$ by Lemma \ref{lemdipolham} we obtain that $H$ is of the form of Equation (\ref{eqtheodipolham}).
\item Suppose that the $y_{0i}$'s are null. Then condition (\ref{eqcond1}) and the fact that the $V_k$'s are linearly independent imply that the matrix $(y_{ij})_{1\leq i,j\leq m}$ is null, and consequently that $Y=0$ (we have assume $y_{00}=0$). It shows that in this case $\Acal(H_1)=\Bcal(\Kcal_1)$, which conclude the proof.
\end{itemize}
\end{proof}

\label{sect13}
\section{The Environment Right-Action Algebra}

If one is only interested in the evolution of the observables of the system, then the algebra $\Acal(U)$ may not be the most relevant. For instance, consider the unitary operators $U$ and $V$ on $\Hcal\otimes\C^2$, written as block-matrices in the canonical basis of $\C^2$:

\begin{equation}
U=\begin{pmatrix}
0 & U_2 \\ U_1 & 0
\end{pmatrix},\quad V=\begin{pmatrix}
U_1 & 0 \\ 0 & U_2\end{pmatrix},
\label{eqexunit1}
\end{equation}

\noindent where $U_1$ and $U_2$ are two unitary operators on $\Hcal$. If $U_1$ and $U_2$ are not collinear, then it is not difficult to compute that $\Acal(U)=\Bcal(\Kcal)$ and $\Acal(V)=\left\lbrace\begin{pmatrix}a & 0 \\ 0 & b\end{pmatrix},\quad a,b\in\C\right\rbrace$ which is a commutative algebra. However, computing the corresponding evolution of an observable of the system $X\in\Bcal(\Hcal)$ in the Heisenberg picture leads to:

\begin{equation*}
U^* \left(X\otimes I_\Kcal\right) U= V^* \left(X\otimes I_\Kcal\right) V.
\end{equation*}

\noindent Thus, from the point of view of the system, $U$  and $V$ have the same action. Consequently $U$ leads to a classical action of the environment in the sense of \cite{ADP2}, a fact which was not captured by the algebra $\Acal(U)$. Remark however that $U$ and $V$ are linked by the relation

\begin{equation}
U=\left[I_\Hcal\otimes\begin{pmatrix}
0 & 1 \\ 1 & 0
\end{pmatrix}\right]V.
\label{eqexunit2}
\end{equation}

We will come back on this remark and generalize it in Subsection \ref{sect22}.

In Subsection \ref{sect21} we introduce the \emph{Environment Right-Action Algebra} and give a first characterization of it. In Subsection \ref{sect23}, we emphasize a link between this algebra and minimal Stinespring representations of CP maps. Finally in Subsection \ref{sect24} we give a characterization of the Environment Algebra and the Environment Right-Action Algebra in term of the spectrum of a CP map, which provide an operational way to compute those two algebras.

\label{sect2}
\subsection{Definition and first examples}

To capture the idea of the action of the environment on the system we introduce the following algebra.

\begin{de}
Let $U$ be a unitary operator on $\Hcal\otimes\Kcal$. Then we call the \emph{Environment Right-Action Algebra} the von Neumann algebra $\Acal_r(U)$ defined by
\begin{equation}
\Acal_r(U)=\left\lbrace U^*(f_1,g_1)U(f_2,g_2);\quad f_1,f_2,g_1,g_2\in\Hcal\right\rbrace''.
\label{eqdeactenvironmentalg}
\end{equation}
\label{deactionenvironment}
\end{de}

\noindent In the same way we could have introduced the \emph{Environment Left-Action Algebra}, in order to describe the relevant part of the environment for the evolution $UX\otimes I_\Kcal U^*$. Most of the results that follow have their counterpart for this algebra, yet we prefer to focus on the Environment Right-Action Algebra as it corresponds to the physical situation of the Heisenberg picture of time evolution.

As for the Environment Algebra, the Environment Right-Action Algebra is easily characterized.

\begin{prop}
Let $U$ be a unitary operator on $\Hcal\otimes\Kcal$. Then $\Acal_r(U)$ is the smallest von Neumann subalgebra of $\Bcal(\Kcal)$ such that
\begin{equation}
U^* \left(X\otimes I_\Kcal\right) U\in\Bcal(\Hcal)\otimes\Acal_r(U)\text{ for all }X\in\Bcal(\Hcal),
\label{eqpropenvironmentact1}
\end{equation}
i.e. if $\Acal$ is an other von Neumann subalgebra of $\Bcal(\Kcal)$ such that Equation (\ref{eqpropenvironmentact1}) holds, then $\Acal_r(U)\subset\Acal$. Furthermore its commutant is given by
\begin{equation}
\Acal_r(U)'=\left\lbrace Y\in\Bcal(\Kcal),\ [I\otimes Y,U^*\left(X\otimes I_\Kcal\right) U]=0\text{ for all }X\in\Bcal(\Hcal)\right\rbrace.
\label{eqpropenvironmentact2}
\end{equation}
\label{propenvironmentact}
\end{prop}

\begin{proof}
The proof is the same as Proposition \ref{propenvironment}. First we prove Equality (\ref{eqpropenvironmentact2}) by a direct computation. Then, if $\Acal$ is an other von Neumann algebra such that Equation (\ref{eqpropenvironmentact1}) holds, one verifies immediately that $\Acal'\subset\Acal_r(U)'$, so that the result follows by the Bicommutant Theorem.
\end{proof}

\begin{coro}
Let $U$ be a unitary operator on $\Hcal\otimes\Kcal$. Then $\Acal_r(U)$ is a subalgebra of $\Acal(U)$.
\label{corodefdecompaction}
\end{coro}

\begin{proof}
As $\Bcal(\Hcal)\otimes I_\Kcal\subset \Bcal(\Hcal)\otimes\Acal(U)$ and $U,U^*\in\Bcal(\Hcal)\otimes\Acal(U)$, we have $U^*\left(\Bcal(\Hcal)\otimes I_\Kcal\right) U\subset \Bcal(\Hcal)\otimes\Acal(U)$. By the characterization of $\Acal_r(U)$, the corollary holds.
\end{proof}

As in the case of the Environment Algebra, commutative Right-Action Algebras and the Spontaneous Emission are the two main examples. The commutative case will be studied as a corollary of Theorem \ref{theoaction}, so that we only treat the case of the Spontaneous Emission here.

\paragraph{Example of the Spontaneous Emission:}
Consider again the unitary operator $U_{\text{se}}$ given by Equation (\ref{eqSpontemiss}), with $\theta\notin 2\pi\Zbb$. We show that $\Acal_r(U)=\Acal(U)=\Bcal(\C^2)$, so that the Environment Right-Action is also purely quantum.
\\To check this, we show that there exists a pure state $\proj{\Omega}{\Omega}$ of $\Kcal$ such that $\Lcal_{\proj{\Omega}{\Omega}}$ is not of the form (\ref{eqexcom2}). If it were the case, $\Lcal_\omega$ would be trace-preserving. Let $(e_0,e_1)$ be the canonical basis of $\C^2$. Then
\begin{align*}
\Lcal_{\proj{e_0}{e_0}}(\proj{e_0}{e_0})
& = \begin{pmatrix}
1 & 0 \\ 0 & \cos{\theta}
\end{pmatrix}\proj{e_0}{e_0}\begin{pmatrix}
1 & 0 \\ 0 & \cos{\theta}
\end{pmatrix} \\
&\quad +\begin{pmatrix}
0 & 0 \\ \sin{\theta} & 0
\end{pmatrix}\proj{e_0}{e_0}\begin{pmatrix}
0 & \sin{\theta} \\ 0 & 0
\end{pmatrix} \\
& = \begin{pmatrix}
1 & 0 \\ 0 & \sin^2{\theta}
\end{pmatrix}.
\end{align*}
It is clear that whenever $\theta\notin 2\pi\Zbb$, $\Tr\left[\Lcal_{\proj{e_0}{e_0}}(\proj{e_0}{e_0})\right]\ne 1$.
\\Consequently, we obtain the announced result:
\begin{equation}
\Acal_r(U_{\text{se}})=\Bcal(\C^2).
\label{eqexspontem2}
\end{equation}
Note that with this method we also directly prove that $\Acal(U_{\text{se}})=\Bcal(\C^2)$, as by corollary \ref{corodefdecompaction} $\Acal_r(U)\subset\Acal(U)$.

\label{sect21}
\subsection{Link between $\Acal(U)$ and $\Acal_r(U)$}

As we remarked before, there is a relation between two unitary operators that have the same action on the system. This point is emphasized in the following lemma (which is originally from \cite{DNP}).

\begin{lem}
Let $U$ and $V$ be two unitary operators on $\Hcal\otimes\Kcal$. Suppose that for all $X\in\Bcal(\Hcal)$,
\begin{equation}
V^* \left(X\otimes I_\Kcal\right) V=U^* \left(X\otimes I_\Kcal\right) U 
\label{eqlemclassaction1}
\end{equation}
Then there exists a unitary operator $W$ on $\Kcal$ such that
\begin{equation}
V=\left(I_\Hcal\otimes W\right) U.
\label{eqlemclassaction2}
\end{equation}
\label{lemclassaction}
\end{lem}

\begin{proof}
From Equation (\ref{eqlemclassaction1}) we obtain that for all $X\in\Bcal(\Hcal)$, $X\otimes I_\Kcal=(VU^*)^* \left(X\otimes I_\Kcal\right) (VU^*)$ and then $[VU^*,X\otimes I_\Kcal]=0$. This in turn implies the existence of a unitary operator $W$ on $\Kcal$ such that $VU^*=I_\Hcal\otimes W$. Clearly $W$ is unitary.
\end{proof}

\noindent We call $\Rcal_r(U)$ the class of $U$ for this relation, that is, $V\in\Rcal_r(U)$ if and only if Equality (\ref{eqlemclassaction1}) holds. It is straightforward that it is an equivalence relation and, because of Proposition \ref{propenvironment}, that every element in this class share the same Environment Right-Action Algebra. That is, for all $V\in\Rcal_r(U)$, we have $\Acal_r(V)=\Acal_r(U)$. Moreover, because of Corollary \ref{corodefdecompaction} for all $V\in\Rcal_r(U)$ one has $\Acal_r(U)\subset\Acal(V)$. Consequently,

\begin{equation*}
\Acal_r(U)\subset \bigcap_{V\in\Rcal_r(U)}\Acal(V). 
\end{equation*}

\noindent A natural question now is whether the equality holds. To answer this question we show under some hypotheses the following: there exists an element in the class $\Rcal_r(U)$ whose Environment Algebra and Environment Right-Action Algebra coincide. This result also implies that we can reduce the study of $\Acal_r(U)$ to the one of $\Acal(V)$ for a specific $V\in\Rcal_r(U)$.

\begin{theo}
Suppose that $\Hcal$ is finite dimensional and that $\Acal_r(U)$ is a type I von Neumann algebra. Then:
\begin{enumerate}
\item there exists a unitary operator $V\in\Rcal_r(U)$ such that:
\begin{equation}
V\in\Bcal(\Hcal)\otimes\Acal_r(U)
\label{eqtheoaction1}
\end{equation}
and consequently
\begin{equation}
\Acal(V)=\Acal_r(V)=\Acal_r(U).
\label{eqtheoaction2}
\end{equation}
\item If $V_1,V_2\in\Rcal_r(U)$ both satisfy Equation (\ref{eqtheoaction2}), then $V_1V_2^*\in I_\Hcal\otimes\Acal_r(U)$.
\end{enumerate}
\label{theoaction}
\end{theo}

At least when $\Kcal$ is finite dimensional, for all $f,g\in\Hcal$, one can write the polar decomposition of the operator $U(f,g)$ as

\begin{equation}
U(f,g)=W\sqrt{U(f,g)^*U(f,g)},
\label{eqpolardecomp}
\end{equation}

\noindent where $W$ is a unitary operator on $\Kcal$ that depends on $f$ and $g$. If it were the case, $\left(I_\Hcal\otimes W\right) U$ would be the wanted element of the class. Theorem \ref{theoaction} can thus be interpreted as a kind of polar decomposition of the operator $U$ with respect to the environment, where the operator $\sqrt{U(f,g)^*U(f,g)}$ is replaced by an element of $\Acal_r(U)$. A similar result can be found in \cite{DNP}.

\begin{proof}
The proof is based on the central decomposition of $\Acal_r(U)$ as a type I von Neumann algebra \cite{Tak}. The space $\Kcal$ has a direct integral representation $\Kcal=\int_A^\oplus {\Kcal_\alpha\ \Prob(d\alpha)}$ in the sense that there exists a family of Hilbert space $(\Kcal_\alpha)_{\alpha\in A}$ and for any $\psi\in\Kcal$ there exists a map $\alpha\in A\mapsto \psi_\alpha$ such that
\begin{equation*}
\sca{\psi}{\phi}=\int_A^\oplus{\sca{\psi_\alpha}{\phi_\alpha}\Prob(d\alpha)}.
\end{equation*}
The von Neumann algebra $\Acal_r(U)$ has a central decomposition
\begin{equation*}
\Acal_r(U)=\int_A^\oplus{\Bcal(\Kcal_\alpha) \Prob(d\alpha)}
\end{equation*}
in the sense that for any $Y\in\Acal_r(U)$ there exists a map $\alpha\in A\mapsto Y_\alpha\in\Bcal(\Kcal_\alpha)$ such that $(Y\psi)_\alpha=Y_\alpha\psi_\alpha$ for almost all $\alpha$. Then for all $X\in \Bcal(\Hcal)$,
\begin{equation*}
U^* \left(X\otimes I_\Kcal\right) U=\int_{\alpha\in A}{T_{\alpha}(X)\Prob(d\alpha)},
\end{equation*}
where $T_\alpha$ are linear maps from $\Bcal(\Hcal)$ to $\Bcal(\Hcal)\otimes\Bcal(\Kcal_\alpha)$. We will ignore for simplicity all issues related with measurability in the following constructions. Let us fix $\alpha\in A$. It is a simple computation to show that $T_\alpha$ is a $*$-morphism that preserves unity: it is a representation of $\Bcal(\Hcal)$ into $\Bcal(\Hcal)\otimes\Bcal(\Kcal_\alpha)$, which furthermore is normal. The goal is to write it in the form $T_\alpha(X)=V_\alpha^* \left(X\otimes I_{\Kcal_\alpha}\right)V_\alpha$ for some unitary operator $V_\alpha$ on $\Hcal\otimes\Kcal_\alpha$.
\\As it is a normal representation of $\Hcal$ on $\Hcal\otimes\Kcal_\alpha$, $\Hcal\otimes\Kcal_\alpha$ can be decomposed as the direct sum of some orthogonal spaces $\Kcal_{\alpha,\beta}$ such that (see \cite{Attb} or \cite{Fagn} for instance)
\begin{equation*}
T_{\alpha}(X)=\sum_{\beta\in B}{V_{\alpha,\beta}XV_{\alpha,\beta}^*},
\end{equation*}
where the $V_{\alpha,\beta}$'s are unitary operators from $\Hcal$ into $\Kcal_{\alpha,\beta}$ and $B$ is a finite or countable set. As $\Hcal$ is finite dimensional, so are the $\Kcal_{\alpha,\beta}$ and they have the same dimension than $\Hcal$. It implies that either $B$ is a finite set with $\dim \Kcal_\alpha$ elements if $\Kcal_\alpha$ is finite dimensional, either $B$ is a countable infinite set if $\Kcal_\alpha$ is infinite dimensional. 
\\Either cases, let $(e_\beta)_{\beta\in B}$ be an orthonormal basis of $\Kcal_\alpha$ and define the operator $V_\alpha$ in $\Bcal(\Hcal\otimes\Kcal_\alpha)$ by the formula
\begin{equation*}
V_\alpha^* \left(f\otimes e_\beta\right) = V_{\alpha,\beta}f,\text{ for all }f\in\Hcal\text{ and }\beta\in B.
\end{equation*}
Then one has 
\begin{align*}
V_\alpha^* V_\alpha
& =\sum_{\beta\in B}{V_\alpha^* \left(I_\Hcal\otimes \proj{e_\beta}{e_\beta}\right)V_\alpha} \\
& = \sum_{\beta\in B}{V_{\alpha,\beta}I_\Hcal V_{\alpha,\beta}^*}\\
& = \sum_{\beta\in B}{I_{\Hcal\otimes\Kcal_{\alpha,\beta}}}=I_{\Hcal\otimes\Kcal_\alpha}.
\end{align*}
Similarly $V_\alpha V_\alpha^*=I_{\Hcal\otimes\Kcal_\alpha}$, so that $V_\alpha$ is a unitary operator. Moreover, for all $f,g\in\Hcal$,
\begin{align*}
V_\alpha^* \left(\proj{f}{g}\otimes I_{\Kcal_\alpha}\right) V_\alpha
& = \sum_{\beta\in B} {V_{\alpha,\beta}\proj{f}{g} V_{\alpha,\beta}^*}\\
& =T_\alpha(\proj{f}{g}).
\end{align*}
This formula extends to all $\Bcal(\Hcal)$ by strong continuity.
\\Now write $V=\int_{\alpha\in A}{V_\alpha\Prob(d\alpha)}$. By construction it is a unitary operator in $\Bcal(\Hcal)\otimes\Acal_r(U)$ and for all $X\in\Bcal(\Hcal)$,
\begin{equation*}
U^* \left(X\otimes I_\Kcal\right) U = V^* \left(X\otimes I_\Kcal\right) V.
\end{equation*}
Now by Lemma \ref{lemunittransfom}, we have $\Acal_r(V)=\Acal_r(U)$ and, as $V\in\Acal_r(V)$, we have $\Acal(V)\subset\Acal_r(V)$ by Proposition \ref{propenvironment}. Consequently $\Acal_r(V)=\Acal(V)$, so point 1) in the Theorem is proved.
Let us prove point 2). If $V_1$ and $V_2$ both satisfy Equation (\ref{eqtheoaction2}), then, again by Lemma \ref{lemunittransfom}, there exists a unitary operator $W$ on $\Kcal$ such that $V_1=I_\Hcal\otimes W V_2$. Then $V_1 V_2^*=I_\Hcal\otimes W$, so that $V_1 V_2^*\in I_\Hcal\otimes\Bcal(\Kcal)$. As $V_1,V_2\in\Bcal(\Hcal)\otimes\Acal_r(U)$, necessarily $W\in\Acal_r(U)$ and the result follows.
\end{proof}

We obtain the following characterization for commutative Environment Right-Action Algebra for a finite dimensional environment:

\begin{coro}
Suppose that $\Hcal$ and $\Kcal$ are finite dimensional, with dimension $N$ and $d$ respectively. Let $U$ be a unitary operator on $\Hcal\otimes\Kcal$. Then $\Acal_r(U)$ is commutative if and only if there exist two orthonormal basis $(\varphi_i)$ and $(\psi_i)$ of $\Kcal$ and unitary operators $U_1,...,U_d$ on $\Hcal$ such that
\begin{equation}
U=\sum_{i=1}^d{U_i\otimes\proj{\varphi_i}{\psi_i}}.
\label{eqcoroaction}
\end{equation}
\label{coroaction}
\end{coro}

\begin{proof}
We first show that (i)$\Rightarrow(ii)$. Because of Theorem \ref{theoaction} there exists a unitary operator $W$ on $\Kcal$ such that
\\$\Acal\left(I_\Hcal\otimes W\ U\right)=\Acal_r(U)$. Consequently, as $\Acal_r(U)$ is commutative, there exist an orthonormal basis $(\psi_i)$ of $\Kcal$ and unitary operators $U_1,...,U_d$ on $\Hcal$ such that
\begin{equation*}
I_\Hcal\otimes W U=\sum_{i=1}^d{U_i\otimes\proj{\psi_i}{\psi_i}}.
\end{equation*}
Write $\varphi_i=W^*\psi_i$ for all $i$. Then $(\varphi_i)$ is an orthonormal basis of $\Kcal$ and Equation (\ref{eqcoroaction}) holds.
\end{proof}

\label{sect22}
\subsection{Minimal Stinespring representation}

There is a link between the structure of $\Acal_r(U)$ and dilation of CP maps obtained via $U$. In particular we obtain a sufficient condition so that $\Acal_r(U)=\Bcal(\Kcal)$. First we recall some definitions.

\begin{de}
Let $\Lcal$ be a normal CP map on $\Hcal$, and $V$ be an isometry from $\Hcal$ to $\Hcal\otimes\Kcal$, such that, for all $X\in\Bcal(\Hcal)$, we have:
\begin{equation*}
\Lcal(X)=V^* \left(X\otimes I_\Kcal\right) V.
\end{equation*}
\\The couple $(\Hcal\otimes\Kcal,V)$ is called a \emph{Stinespring representation} of $\Lcal$. Furthermore, this representation is defined to be \emph{minimal} if the set
\begin{equation}
\Vcal=\left\{(X\otimes I) U f\otimes \psi,\ X\in\Bcal(\Hcal),f\in\Hcal\right\}
\label{eqdedilmin}
\end{equation}
\noindent is total in $\Hcal\otimes\Kcal$.
\label{dedilmin}
\end{de}

\noindent By Stinespring Theorem (\cite{Sti}) every normal CP map has a minimal Stinespring representation of this form (see \cite{Attb} and \cite{Fagn}). Furthermore, if $(\Hcal\otimes\Kcal_1,V_1)$ and $(\Hcal\otimes\Kcal_2,V_2)$ are two minimal Stinespring representations of $\Lcal$, then there exists a unitary operator $W$ from $\Kcal_1$ to $\Kcal_2$ such that for all $X\in\Bcal(\Hcal)$:

\begin{equation*}
\left(I_\Hcal\otimes W\right) V_1 = V_2.
\end{equation*}

\noindent Now we come back to the case of a unitary operator $U$ on the bipartite space $\Hcal\otimes\Kcal$. For every pure state $\psi\in\Kcal$ define the isometry $U_\psi:\Hcal\to\Hcal\otimes\Kcal$ by

\begin{equation}
U_\psi:f\mapsto U f\otimes \psi.
\label{equnitarydil}
\end{equation}

\noindent We remark that the couple $(\Hcal\otimes\Kcal,U_\psi)$ is a Stinespring representation of the CP map $\Lcal_{\proj{\psi}{\psi}}$ defined by

\begin{equation}
\Lcal_{\proj{\psi}{\psi}}:X\in\Bcal(\Hcal)\to\Tr_{\proj{\psi}{\psi}}\left[U^*X\otimes I_\Kcal U\right].
\end{equation}

\noindent A direct application of Lemma \ref{lemclassaction} is that $V\in\Rcal_r(U)$ if and only if $\Tr_{\proj{\psi}{\psi}}\left[V^*\cdot\otimes I_\Kcal V\right]=\Lcal_{\proj{\psi}{\psi}}(\cdot)$ for all pure state $\psi\in\Kcal$ (which was the initial Lemma in \cite{DNP}). The following proposition completes this statement.

\begin{prop}
The property for $(\Hcal\otimes\Kcal,U_\psi)$ to give a minimal Stinespring representation of $\Lcal_{\proj{\psi}{\psi}}$ is a property of the class $\Rcal_r(U)$, i.e. if $(\Hcal\otimes\Kcal,U_\psi)$ is minimal, then for all $V\in\Rcal_r(U)$, the Stinespring representation $(\Hcal\otimes\Kcal,V_\psi)$ is also minimal.
\label{propminimaldil}
\end{prop}

\begin{proof}
Let $\psi\in\Kcal$ be a pure state. For all unitary operators $V$ on $\Kcal$, we write
\begin{equation*}
\Vcal_\psi(V)=\left\{(X\otimes I) V f\otimes\psi,\ X\in\Bcal(\Hcal),f\in\Hcal\right\}.
\end{equation*}
Suppose that $(\Hcal\otimes\Kcal,U_\psi)$ is a minimal Stinespring representation of $\Lcal_{\proj{\psi}{\psi}}$. By definition it means that $\Vcal_\psi(U)$ is total in $\Hcal\otimes\Kcal$. Take $V\in\Rcal_r(U)$. Because of Lemma \ref{theoaction}, there exists a unitary operator $W$ on $\Kcal$ such that $V=\left(I_\Hcal\otimes W\right) U$. Let $x$ be an element of $\Vcal_\psi(V)^\perp$. Then for all $X\in\Bcal(\Hcal)$ and $f\in\Hcal$:
\begin{equation*}
0 = \sca{x}{(X\otimes W)U\left(f\otimes\psi\right)} = \sca{\left(I\otimes W^*\right)x}{(X\otimes I)U\left(f\otimes\psi\right)}.
\end{equation*}
\noindent As $\Vcal_\psi(U)$ is total in $\Hcal\otimes\Kcal$, this implies that $\left(I\otimes W^*\right)x=0$, and consequently $x=0$.
\end{proof}

Let $\Acal$ be a von Neumann subalgebra of $\Bcal(\Kcal)$. We recall that a vector $\psi\in\Kcal$ is cyclic for $\Acal$ if $\overline{\Acal\psi}=\Kcal$. There does not always exist a cyclic vector. For example, a commutative algebra $\Acal$ has a cyclic vector if and only if $\Acal=\Acal'$ (\cite{R-S}).

We now suppose that $\Hcal$ is finite dimensional and that $\Acal_r(U)$ is a type I von Neumann algebra in order to be in position to apply Theorem \ref{theoaction}

\begin{prop}
If $(\Hcal\otimes\Kcal,U_\psi)$ is a minimal Stinespring representation of $\Lcal_{\proj{\psi}{\psi}}$, then $\psi$ is a cyclic vector for $\Acal_r(U)$.
\label{propcyclic}
\end{prop}

\begin{proof}
First, by Theorem \ref{theoaction}, there exists $V\in\Rcal_r(U)$ such that $\Acal(V)=\Acal_r(U)$. Recall that $\Acal(V)$ is generated by the elements $V(f,g)$ defined by Equation (\ref{eqU(f,g)}), $f,g\in\Hcal$. Consequently the subset of $\Kcal$
\begin{equation*}
\Kcal_\psi(V)=\left\lbrace V(f,g)\psi,\ f,g\in\Hcal\right\rbrace,
\end{equation*}
is a subset of $\Acal_r(U)\psi$. Now if $(\Hcal\otimes\Kcal,U_\psi)$ is a minimal Stinespring representation of $\Lcal_{\proj{\psi}{\psi}}$, it is also the case for $(\Hcal\otimes\Kcal,V_\psi)$ by Proposition \ref{propminimaldil}, so that $\Vcal_\psi(V)$ is a total set. Let us show that $\Kcal_\psi(V)$ is also a total set, to complete the proof. Take $\varphi\in\Kcal_\psi(V)^\perp$. Then for all $f,g\in\Hcal$ we have:
\begin{align*}
0 & = \sca{\varphi}{V(f,g)\psi} \\
& = \Tr\left[\proj{\psi}{\varphi}V(f,g)\right] \\
& = \Tr\left[\left(\proj{g}{f}\otimes\proj{\psi}{\varphi}\right)V\right] \\
& = \sca{f\otimes\varphi}{Vg\otimes\psi}.
\end{align*}
Consequently, for all $f\in\Hcal$, $f\otimes\varphi\in\Vcal_\psi^\perp=\{0\}$. Thus $\varphi=0$, which concludes the proof.
\end{proof}

\noindent The following corollary is now straightforward.

\begin{coro}
If for all $\psi\in\Kcal$, $(\Hcal\otimes\Kcal,U_\psi)$ is a minimal Stinespring representation of $\Lcal_{\proj{\psi}{\psi}}$, then $\Acal_r(U)=\Bcal(\Kcal)$.
\label{corocyclic}
\end{coro}

\begin{proof}
It just comes from the fact that if $\Acal$ is a subalgebra of $\Bcal(\Kcal)$ such that $\Acal\psi$ is total for all $\psi\in\Kcal$, then $\Acal=\Bcal(\Kcal)$.
\end{proof}

\label{sect23}
\subsection{A characterization of the algebras in terms of the spectrum of a CP map}

In this Subsection, $\Hcal$ and $\Kcal$ are finite dimensional Hilbert spaces, of dimension $N$ and $d$ respectively. We give a characterization of the algebras $\Acal(U)$ and $\Acal_r(U)$ in terms of the spectrum of a specific CP map $\Lcal$ acting on $\Bcal(\Kcal)$. This CP map describes the evolution of the states of $\Kcal$, i.e. density matrices on $\Kcal$, whenever the state of the system is the maximally mixed state $\frac{1}{N}I_\Hcal$. Thus $\Lcal$ is given by the following formula, where $X\in\Bcal(\Kcal)$:

\begin{equation}
\Lcal(X)=\Tr_\Hcal\left[U\left( \frac{1}{N}I_\Hcal\otimes X\right) U^*\right].
\label{eqCP1}
\end{equation}

\noindent We will also need the adjoint of $\Lcal$ for the Hilbert-Schmidt scalar product on $\Bcal(\Kcal)$ $(X,Y)\mapsto \Tr[X^*Y]$:

\begin{equation}
\Lcal^*(X)=\Tr_\Hcal\left[U^*\left( \frac{1}{N}I_\Hcal\otimes X\right) U\right].
\label{eqCP2}
\end{equation}

\noindent Our goal is to give an explicit way to compute the algebras, by relating them to some eigenspaces. Our result is the following:

\begin{theo}
Suppose that $\Hcal\approx\C^N$ and $\Kcal\approx\C^{d}$. Then:
\begin{enumerate}
\item[1)] $\Acal(U)'$ is the eigenspace of both $\Lcal$ and $\Lcal^*$, associated to the eigenvalue $1$;
\item[2)] $\Acal_r(U)'$ is the right-singularspace of $\Lcal$ associated to the right-eigenvalue $1$.
\end{enumerate}
\label{theocharaspect}
\end{theo}

The proof of Theorem \ref{theocharaspect} will make use of Lemma \ref{lementropy} below, in which we use the notions of Shannon entropy and von Neumann entropy. The Shannon entropy of a discrete probability measure $(p_i)$ is defined by:

\begin{equation}
H(p_i)=-\sum_i{\phi(p_i)},
\label{eqShannonentropy}
\end{equation}

\noindent where $\phi$ is the real and operator-monotone function $\phi:x\mapsto x\log x$, defined on $[0,1]$. The von Neumann entropy of a density matrix $\omega$ on $\Kcal$ is the quantity

\begin{equation}
\Scal(\omega)=-\Tr [\phi(\omega)].
\label{eqVNentropy}
\end{equation}

\begin{lem}
For all density matrix $\omega$ on $\Kcal$,
\begin{equation}
\Scal(\omega)\leq\Scal(\Lcal(\omega)),
\label{eqlementropy1}
\end{equation}
i.e. $\Lcal$ is entropy increasing. Furthermore equality holds if and only if
\begin{equation}
U\left( \frac{1}{N}I_\Hcal\otimes\omega\right) U^*=\frac{1}{N}I_\Hcal\otimes\Lcal(\omega).
\label{eqlementropy2}
\end{equation}
\label{lementropy}
\end{lem}

\noindent Inequality (\ref{eqlementropy1}) is not new. In fact $\Lcal$ is a doubly stochastic CP map (i.e. it is identity and trace preserving) and basic results in majorization theory imply that a CP map verify Inequality (\ref{eqlementropy1}) if and only if it is doubly stochastic (\cite{MOA}). However, as we are interested in the equality case, it is necessary for us to prove it.

\begin{proof}[Proof of Lemma \ref{lementropy}]
We will use the three following properties of Shannon entropy. Let $(P_i)_{i\in I}$ be a complete set of orthogonal projection on $\Kcal$. Let $\omega$ be a density matrix on $\Kcal$ and write $\omega'=\sum_{i\in I}{P_i\omega P_i}$. Let $(p_i)_{1\leq i \leq m}$ be a probability measure on $\{1,...,m\}$ for some $m$ and $\omega_i$ be density matrices on $\Kcal$. Then:
\begin{enumerate}
\item $\Scal(\omega)\leq\Scal(\omega')$, with equality if and only if $\omega=\omega'$;
\item $\sum_{i=1}^m{p_i\Scal(\omega_i)}\leq\Scal(\sum_i{p_i\omega_i})$, with equality if and only if the $\omega_i$'s are identical;
\item $\Scal(\sum_i{p_i\omega_i})\leq\sum_i{p_i\Scal(\omega_i)}+H(p_i)$, with equality if and only if the $\omega_i$'s have orthogonal supports.
\end{enumerate}
Let $(e_i)_{1\leq i \leq N}$ be an orthonormal basis of $\Hcal$ and write $P_i=\proj{e_i}{e_i}\otimes I_\Kcal$ the orthogonal projection on the subspace $\C e_i\otimes \Kcal$. The $P_i$'s form a complete set of orthogonal projections on $\Kcal$. Define
\begin{align*}
& p_i=\Tr\left[P_i U\left( \frac{1}{N}I_\Hcal\otimes\omega\right) U^* \right]; \\
& \tilde{\omega}_i=\frac{1}{p_i}P_i U\left( \frac{1}{N}I_\Hcal\otimes\omega\right) U^* P_i; \\
& \omega_i=\Tr_\Hcal[\tilde{\omega}_i]\text{ so that }\tilde{\omega}_i=\frac{1}{p_i}\proj{e_i}{e_i}\otimes\omega_i.
\end{align*}
Remark that $(p_i)$ defines a probability measure on $\{1,...,N\}$ and that the $\tilde{\omega}_i$'s have orthogonal supports. Furthermore, it is not difficult to check that $\Scal(\omega_i)=\Scal(\tilde{\omega}_i)$ for all $i=1,...,N$ and $\Lcal(\omega)=\sum_i{p_i\omega_i}$. We are now ready to make the following computation:
\begin{align*}
\log N + \Scal(\omega)
& = \Scal(\frac{1}{N}I_\Hcal\otimes\omega) \\
& = \Scal(U\ \frac{1}{N}I_\Hcal\otimes\omega\ U^*) \\
& \leq \Scal(\sum_{i=1}^N{p_i\tilde{\omega}_i})\text{ because of 1.}\\
& = \sum_{i=1}^N{p_i\Scal(\tilde{\omega}_i)}+H(p_i)\text{ because of 3.} \\
& = \sum_{i=1}^N{p_i\Scal(\omega_i)}+H(p_i) \\
& \leq \Scal(\sum_{i=1}^N{p_i\omega_i}) + H(p_i)\text{ because of 2.}\\
& \leq \Scal(\Lcal(\omega)) + \log N.
\end{align*}
Consequently, we obtain the desired inequality:
\begin{equation*}
\Scal(\omega)\leq\Scal(\Lcal(\omega)).
\end{equation*}
The Equality case above means that all the previous inequalities are in fact equalities. In particular:
\begin{enumerate}
\item[(i)]$\Scal\left(U\left(\frac{1}{N}I_\Hcal\otimes\omega\right)U^*\right)=\Scal(\sum_{i=1}^N{p_i\tilde{\omega}_i})$,
\item[(ii)]$H(p_i)=\log N$, 
\item[(iii)]$\sum_{i=1}^N{p_i\Scal(\omega_i)}=\Scal(\sum_{i=1}^N{p_i\omega_i})$.
\end{enumerate}
Consequently by 1. and 2. and the fact that $H(p_i)=\log N$ if and only if $(p_i)$ is the uniform probability measure, we get:
\begin{enumerate}
\item[(i)]$U\left(\frac{1}{N}I_\Hcal\otimes\omega\right)U^*=\sum_i{p_i\tilde{\omega}_i}$,
\item[(ii)]$p_i=\frac{1}{N}$ for all $i=1,...,N$,
\item[(iii)]$\omega_i=\omega'$ for all $i=1,...,N$ for some density matrix $\omega'$ on $\Kcal$.
\end{enumerate}
Using those three points, we obtain
 \begin{equation*}
U\ \left(\frac{1}{N}I_\Hcal\otimes\omega\right)\ U^*= \sum_i{\frac{1}{N}\proj{e_i}{e_i}\otimes\omega'}=\frac{1}{N}I_\Hcal\otimes\omega'.
\end{equation*}
Taking the partial trace with respect to $\Hcal$, we see that $\omega'=\Lcal(\omega)$ which end the proof.
\end{proof}

\begin{proof}[Proof of Theorem \ref{theocharaspect}]
The idea behind the proof of each part of the theorem is the same. We start with the following computation. For all $X\in\Bcal(\Kcal)$; we have
\begin{align*}
X\in\Acal'(U) 
& \Leftrightarrow [I_\Hcal\otimes X,U]=0 \\
& \Leftrightarrow U \left(\frac{1}{N}I_\Hcal\otimes X\right) U^*= \frac{1}{N}I_\Hcal\otimes X \\
& \qquad\text{ and }\quad U^* \left(\frac{1}{N}I_\Hcal\otimes X\right) U= \frac{1}{N}I_\Hcal\otimes X \\
& \Rightarrow \Lcal(X)=X\text{ and }\Lcal^*(X)=X \\
& \Rightarrow X \text{ is an eigeinvector of }\\
&\qquad\Lcal\text{ associated to the eigeinvalue }1.
\end{align*}
In the same way, for all $X\in\Bcal(\Kcal)$
\begin{align*}
X\in\Acal_r'(U) 
& \Leftrightarrow [I_\Hcal\otimes X,U^* \left(Y\otimes I_\Kcal\right)]=0\quad\forall Y\in\Bcal(\Kcal) \\
& \Leftrightarrow [U \left(I_\Hcal\otimes X\right) U^*, Y\otimes I_\Kcal]=0\quad\forall Y\in\Bcal(\Kcal) \\
& \Leftrightarrow \text{ there exists }X'\in\Bcal(\Kcal)\text{ such that:}\\
& \qquad U\left( \frac{1}{N}I_\Hcal\otimes X \right)U^*= \frac{1}{N}I_\Hcal\otimes X'\\
& \Rightarrow  \text{ there exists }X'\in\Bcal(\Kcal)\text{ such that:}\\
& \qquad\Lcal(X)=X'\text{ and }\Lcal^*(X')=X \\
& \Rightarrow X \text{ is a right-singularvector of }\\
& \qquad\Lcal\text{ associated to the singularvalue }1.
\end{align*}
In order to complete the proof, we need to show the converse. We do that for $\Acal_r'(U)$ only, as it is the same for $\Acal(U)'$.
\\First remark that $\Lcal(I_\Kcal)=\Lcal^*(I_\Kcal)=I_\Kcal$, so that $I_\Kcal$ is always an eigenvector of $\Lcal$ and $\Lcal^*$ for the eigenvalue $1$ (and consequently a left and right singularvector for the singularvalue $1$). Furthermore, if $\Lcal(X)=X'$ and $\Lcal^*(X')=X$ for some $X,X'\in\Bcal(\Kcal)$, than $\Lcal(X^*)={X'}^*$ and $\Lcal^*({X'}^*)=X^*$. Consequently, the right-singularspace of $\Lcal$ associated to the right-eigenvalue $1$ is a system operator, that is a norm-closed $*$-stable subspace of $\Bcal(\Kcal)$. In particular, it is generated by its positive elements, so that we only need to prove the result for density matrices.
\\In view of Lemma \ref{lementropy}, we only need to prove that for all density matrices $\omega$ such that $\omega=\Lcal^*\circ\Lcal(\omega)$,
\begin{equation*}
\Scal(\omega)=\Scal(\Lcal(\omega)).
\end{equation*}
\noindent Applying again Lemma \ref{lementropy} both for $U^*$ and $U$, we get
\begin{equation*}
\Scal(\Lcal^*\circ\Lcal(\omega))=\Scal(\omega)\leq\Scal(\Lcal(\omega))\leq\Scal(\Lcal^*\circ\Lcal(\omega)),
\end{equation*}
\noindent which shows the desired equality.
\end{proof}

\label{sect24}

\bibliographystyle{abbrv}
\bibliography{biblio}
\end{document}